\documentclass[letterpaper, 10pt, conference]{IEEEtran} 
\normalsize

\DeclareMathAlphabet\mathbfcal{OMS}{cmsy}{b}{n}

\usepackage{amsthm}
\usepackage{bbm}
\usepackage{algorithm}
\usepackage{algorithmic}
\usepackage{graphicx,amsmath,amssymb,latexsym,epsfig,url}

\DeclareMathOperator*{\esssup}{ess\,sup}
\usepackage{setspace}
\usepackage{psfrag}
\usepackage{array,booktabs,arydshln,xcolor,cite}

\usepackage{lscape}
\newtheorem{theorem}{Theorem}

\newtheorem{corollary}{Corollary}
\newtheorem*{proof of Theorem*}{Proof of Theorem 3}
\newtheorem{proof of Lemma}{Proof of Lemma}
\newtheorem{definition}{Definition}
\newtheorem{lemma}{Lemma}

\begin{document}
\vspace{-0.2 in}
\title{Achieving the Age-Energy Tradeoff with a Finite-Battery Energy Harvesting Source}
\author{\IEEEauthorblockN{Baran Tan Bacinoglu\IEEEauthorrefmark{1},  Yin Sun\IEEEauthorrefmark{3}, Elif Uysal-Biyikoglu\IEEEauthorrefmark{1}, and Volkan Mutlu \IEEEauthorrefmark{1}}
\IEEEauthorblockA{\IEEEauthorrefmark{1}METU, Ankara, Turkey,
\IEEEauthorrefmark{3}Auburn University, AL, USA\\
 E-mail:  barantan@metu.edu.tr, yzs0078@auburn.edu, uelif@metu.edu.tr,  volkan.mutlu@metu.edu.tr}
}

\bibliographystyle{IEEEtran}

\maketitle
\vspace{-0.1 in}
\def\eg{\emph{e.g.}}
\def\ie{\emph{i.e.}}

\begin{abstract}
 We study the problem of minimizing the time-average expected Age of Information for status updates sent by an energy-harvesting source with a finite-capacity battery. In prior literature, optimal policies were observed to have a threshold structure under Poisson energy arrivals, for the special case of a unit-capacity battery. In this paper, we generalize this result to any (integer) battery capacity, and explicitly characterize the threshold structure. We obtain tools to derive the optimal policy for arbitrary energy buffer (i.e. battery) size. One of these results is the unexpected equivalence of the minimum average AoI and the optimal threshold for the highest energy state.
\end{abstract}
\begin{IEEEkeywords}Age of Information; age-energy tradeoff; threshold policy; optimal threshold; energy harvesting; battery capacity\end{IEEEkeywords}
\section{Introduction}
The Age of Information (AoI) was proposed in \cite{ Kaul2011, Kaul2012} as a performance metric that measures the freshness of information in status-update systems. For a flow of information updates sent from a source  to a destination, status age is defined as the time elapsed since the newest update available was generated at the source. That is, if $u(t)$ is the largest among the time-stamps of all packets received  by time $t$, status age is defined as:
\begin{equation}
\Delta(t)= t-u(t),
\end{equation}
The \emph{Age of Information} (AoI) usually refers to the time-average of $\Delta(t)$. AoI is a particularly relevant performance metric for status-update applications that have growing importance in social networks, remote monitoring \cite{ZviedrisESMS10, blueforce}, machine-type communication (smart cities, industrial manufacturing, telerobotics, IoT).

AoI was analyzed under various queueing system models, service disciplines and queue management policies in recent literature (\eg, \cite{Ephremides2013, Ephremides2014, Huang2015, Pappas2015, Ephremides2016, Najm2016, DBLP:journals/corr/YatesK16, StatusUpdateHARQ}.  The control and optimization of AoI for an active source that can generate updates at will, was studied in \cite{YinSunInfocom2016, YinIT2017}. 
     
The relation of energy and AoI was studied as early as 2015:  The problem of AoI-optimal generation of status updates when the source is constrained by an arbitrary sequence of energy arrivals was formulated in~\cite{TanITA2015}, resulting in the optimal offline solution and an online policy. The study in \cite{2015ISITYates} considered the optimization of AoI under a long-term average rate of energy harvesting, when update transmissions are subject to random delays in the network. Both studies observed that AoI-optimal policies tend to be \emph{lazy}, in the sense that they may intentionally impose a waiting time before sending the next update. That is, for maximum freshness, one may sometimes send updates at a rate lower than one is allowed to- which may be counter-intuitive at first sight. 

The online problem in \cite{TanITA2015} was extended to a continuous-time formulation with Poisson energy arrivals, finite energy storage (battery) capacity, and random packet errors in the channel in~\cite{TanISIT2017}. An age-optimal threshold policy was proposed for the unit battery case, and the achievable AoI for arbitrary battery size was bounded  for a channel with a constant error packet error probability. Optimal threshold policies for the unit battery and infinite battery capacity cases were found for a channel with no errors, in the concurrent study in~\cite{WuYang2017}. The problem of characterizing optimal policies for arbitrary battery sizes remained open.  

In \cite{twohop_energyharvesting}, the offline results in \cite{TanISIT2017} were extended considering fixed non-zero service time and the result is used to obtain a solution for the two-hop scenario. Another offline problem under energy harvesting was investigated in \cite{ArafaDelay} where the transmission delay of an update is controlled by energy consumed on its transmission.  

This paper extends~\cite{TanISIT2017}, making the following contributions:
\begin{itemize}
\item A more general description of the policy space for age-optimal scheduling, including threshold policies with age-based thresholds that are monotone in energy state, is formulated.
\item Following the study in \cite{TanISIT2017}, it was conjectured that for any battery size, the optimal threshold on the age for the highest energy state is actually equal to the minimum AoI. This conjecture is proved to be correct.
\item Optimal thresholds are obtained  numerically for integer battery size up to 5. 
\end{itemize}
\section{System Model}
\label{sec:sm}
Consider an energy harvesting transmitter that sends update packets to a destination, as illustrated in Fig \ref{sensormodelw}. Suppose that the transmitter has a finite battery which is capable of storing up to $B$ units of energy. Transmission of  each update packet consumes a unit of energy. Let $E(t)$ denote the amount of energy stored in the battery at time $t$. 
The timing of status updates are controlled by a sampler which can monitor the battery level $E(t)$ at all time $t$.  
\begin{figure}[htpb]
    \centering \includegraphics[scale=0.72]{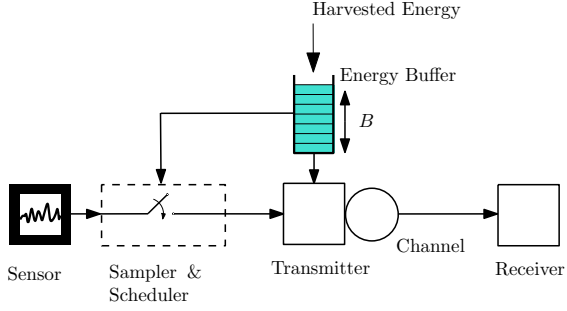}
\caption{System Model.}
\label{sensormodelw} 
\end{figure}
We assume that when an update is given to the transmitter, it is instantaneously transmitted \footnote{This corresponds to an assumption of instantaneous service, i.e, the duration of packet transmission is ignored. This is an appropriate model for sporadic transmissions (\eg, a sensor reporting temperature) where the time between two updates is typically much larger than a packet transmission duration.}. 
 
Let $N_H(t)$ and $N_U(t)$ denote the number of energy units that have arrived and the number of updates that have sent out by time $t$, respectively. We assume that the energy arrival process is Poisson with a rate $\mu_{H}$.  Energy arriving while the battery is full is lost (cannot be stored or used).

The system starts to operate at time $t=0$. Let $Z_k$ denote the generation time of the $k$-th update packet such that $0=Z_0 \leq Z_1\leq Z_2\leq\ldots$. An update policy is defined by a sequence of update instants $\pi=(Z_{0},Z_{1},Z_{2},...)$. In many status-update systems  (e.g., a sensor reporting temperature \cite{}), the update packets are only sent out sporadically and the packet size is quite small. Hence, the duration for transmitting a packet is much smaller than the difference between two subsequent update times. Motivated by this, we assume that the packet transmission time can be approximated as zero. With this assumption, the age at a status generation is zero, i.e. $\Delta(Z_{k})=0$ for any $k$, and the age at any time $t$ is:
\begin{equation}
\label{efuage}
\Delta(t)=t-Z_{N_U(t)}, t\geq 0.
\end{equation}

The battery level before the $(k+1)$-st update instant is given by the following:
\begin{align}
\label{energyavailability}
&E(t)=\min \lbrace (E(Z_{k})-1)^{+}+N_{H}(t)-N_{H}(Z_{k}),B\rbrace, &\nonumber\\
& t \in (Z_{k}, Z_{k+1}],&
\end{align}

\noindent{We first define the set of \emph{energy-causal} update policies:}
\begin{definition}
A policy $\pi$ is said to be energy-causal if no update packet is sent out when the battery is empty, i.e., $E(Z_{k})\geq 1$ for all $k \geq 1$.
\end{definition}
The information available up to some time $t$ is represented by $\mathcal{F}_{t}=\sigma(\lbrace( N_H(t'),N_U(t')),0 \leq t' < t  \rbrace)$ which is the $\sigma$-field generated by the sequence of energy arrivals and updates, i.e., $\lbrace( N_H(t'),N_U(t')),0 \leq t' < t  \rbrace$. The set of  \emph{online} update policies is defined  as follows:
\begin{definition}
An energy-causal policy is said to be online if no update instant is determined based on future information, i.e., does not depend on future events, i.e.,  $\left\lbrace Z_{k} \leq t\right\rbrace \in \mathcal{F}_{t}$ for all $t\geq 0$ and $k \geq 1$.
\end{definition}
Let $\Pi$ denote  the set of online update policies.
The time-average expected age can be expressed as:
\begin{equation}
\label{avaged}
\bar{\Delta}=\displaystyle\lim\sup_{t_{f} \rightarrow \infty}\frac{1}{t_{f}}\mathbb {E}\left[ \displaystyle\int_{0}^{t_{f}}\Delta(t)  dt \right].
\end{equation}
Let $X_{k}$ represent the inter-update duration between updates $k-1$ and $k$, i.e., $X_{k}=Z_{k}-Z_{k-1}$. Then, the time-average expected age in (\ref{avaged}) can be equivalently expressed as:
\begin{equation}
\label{avaged2}
\bar{\Delta}=\displaystyle\lim\sup_{t_{f} \rightarrow \infty}\frac{1}{2t_{f}}\mathbb {E}\left[ \displaystyle\sum_{k=1}^{N_{U}(t_{f})} X_{k}^{2}+(t_{f}-Z_{N_{U}(t_{f})})^{2}\right].
\end{equation} 
The goal of this paper is to find the optimal update policy for minimizing the time-average expected age, which is formulated as:
\vspace{-0.05 in}
\begin{equation}
\label{minavaged}
\displaystyle\min_{\pi\in\Pi}\bar{\Delta}.
\end{equation}
\section{MAIN RESULTS}
\label{sec:mainresults}

We begin with a result guaranteeing the existence of threshold-type policies that are optimal. We define such policies as follows:
\begin{definition}
\label{monotonicthresholds}
An online policy is said to be a  threshold policy if:
\begin{equation}
\label{thresholdpolicy}
Z_{k+1} =\inf \left\lbrace t \geq Z_{k} : \Delta(t) \geq \tau_{E(t)} \right\rbrace ,
\end{equation}
where $\tau_{\ell}$ denotes the threshold for sending an update when the battery level is $\ell$ for $\ell=1,\ldots, B$.
\end{definition}
Let $\Pi^{T}\subset \Pi$ be the set of threshold policies. First, we note the following:
\begin{theorem}
\label{existopthreshold}
There exists a threshold policy $\pi \in \Pi^{T}$ that solves (\ref{minavaged}).
\end{theorem}
In our search for an optimal policy, we can reduce the space of policies further,
\begin{definition}
\label{monotonicthresholds}
A threshold policy is said to be a monotone threshold policy if $\tau_B\leq \ldots\leq\tau_{\ell}\leq\tau_{1}$.
\end{definition}
Let $\Pi^{\rm{MT}}$ be the set of monotone threshold policies. The following is true: 
\begin{theorem}
\label{existopmothreshold}
There exists a monotone threshold policy $\pi \in \Pi^{\rm{MT}}$ that solves (\ref{minavaged}).
\end{theorem}
Theorem \ref{existopmothreshold} implies that in the optimal update policy, update packets are sent out more frequently when the battery level is high. 
\begin{figure}[htpb]
    \centering \includegraphics[scale=0.72]{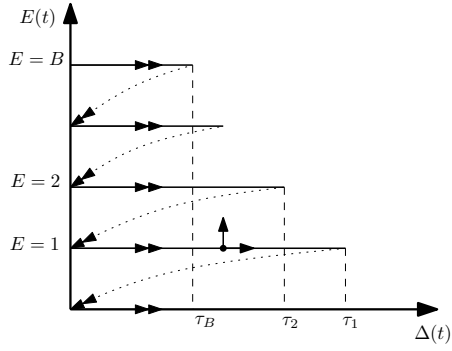}
\caption{An illustration of the state space and transitions for a policy in $\Pi^{\rm{MT}}$.}
\label{statespacew} 
\end{figure}

To understand the time evolution of $\Delta(t)$ and $E(t)$ for policies in $\Pi^{\rm{MT}}$, consider the illustration in Fig.  
\ref{statespacew}. It can be seen from Fig.  
\ref{statespacew} that when $Z_{k}^{+}=j$, the next update of a policy $\pi \in \Pi^{\rm{MT}}$ occurs before than some time $t' \in [Z_{k}+\tau_{m}, Z_{k}+\tau_{m-1}]$ if and only if there occur $m-j$ energy arrivals before than $t'$. Accordingly, for  policies in $\Pi^{\rm{MT}}$, the cumulative distribution function (CDF) of inter-update durations, $\Pr(X_{k+1}\leq x \mid E(Z_{k}^{+})=j)$ can be expressed  as:
\begin{eqnarray}
&\Pr(X_{k+1}\leq x \mid  E(Z_{k}^{+})=j)  = & \nonumber \\
&\begin{cases} 0, &\text{if~} x < \tau_{B} \\
\Pr(Y_{m-j}\leq x), &\text{if~}  \tau_{m}\leq x < \tau_{m-1}, \forall m \in \lbrace 2,...,B\rbrace \\
\Pr(Y_{1-j}\leq x),  &\text{if~} \tau_{1}\leq x 
\end{cases}&
\label{cdferlang}
\end{eqnarray}

where $Y_{i}$ obeys the Erlang distribution at rate $\mu_H$ with parameter $i$, for $i\geq 1$, and $Y_{i}=0$ for $i \leq 0$.  From (\ref{cdferlang}), an expression for the transition probability $\Pr(E(Z_{k+1}^{+})=i \mid  E(Z_{k}^{+})=j)$ for $i= 0,1, ...., B-1$ can be derived:
\begin{eqnarray}
&\Pr(E(Z_{k+1}^{+})=i \mid  E(Z_{k}^{+})=j)=& \nonumber \\
&\begin{cases} \Pr(Y_{B-j}\leq \tau_{B-1}), &\text{if~} i=B-1 \\
\Pr(Y_{1+i-j}\leq \tau_{i})- \Pr(Y_{2+i-j}\leq \tau_{i+1}), &\text{if~}  i <B-1 \\
\end{cases}&
\label{transitionerlang}
\end{eqnarray}

Hence, energy states sampled at update instants can be described as a DTMC with the transition probabilities in (\ref{transitionerlang}). When thresholds are finite, this DTMC is ergodic as any energy state is reachable from any other energy state with positive probability in $B-1$ steps. 

Next, we show the main structural result satisfied by the thresholds of any optimal policy in $\Pi^{\rm{MT}}$.
\begin{theorem}
\label{fixedpthreshold}
An optimal policy for solving (\ref{minavaged}) is a monotone threshold policy  $\pi^* = (\tau_1^*, \ldots, \tau_B^*)$ that satisfies the following property: The threshold $\tau^*_B$ for sending an update packet when the battery is full is equal to the minimum time-average expected age, i.e.,
\begin{equation}
\tau^*_B = \bar{\Delta}_{\pi^*} =  \min_{\pi\in\Pi} \bar{\Delta}_\pi. 
\end{equation}
\end{theorem}
This follows from the following two results:
\begin{lemma}
\label{diffmoments}
Consider non-negative random variable $X$, if:
\begin{eqnarray*}
&\Pr(X\leq x)  = & \nonumber \\
&\begin{cases} 0 &\mbox{; } x < \tau_{B}, \\
F_{i}(x) & \mbox{; }  \tau_{i}\leq x < \tau_{i-1}, \forall i \in \lbrace 2,...,B \rbrace , \\
F_{1}(x)  & \mbox{; } \tau_{1}\leq x, 
\end{cases}&
\end{eqnarray*}
where $\tau_{B}\leq ... \leq \tau_{2} \leq \tau_{1}$ and $F_{i}(x)$ is the CDF of a non-negative random variable for every $i \in \lbrace 1,...,B \rbrace$, then:
\begin{equation*}
\frac{\partial}{\partial \tau_{i}}\mathbb {E}\left[ X^{2} \right]=2\tau_{i}\frac{\partial}{\partial \tau_{i}}\mathbb {E}\left[ X \right].
\end{equation*}
\end{lemma}
\begin{corollary}
The inter-update intervals, $X$, for any $\pi \in \Pi^{\rm{MT}}$ satisfy the following: 
\begin{equation}
\frac{\partial}{\partial \tau_{i}}\mathbb {E}\left[ X^{2} \mid j\right]=2\tau_{i}\frac{\partial}{\partial \tau_{i}}\mathbb {E}\left[ X \mid j\right], \forall (i,j)\in \lbrace 1, 2,...,B\rbrace^{2}.
\end{equation}
\end{corollary}
Note that the transition probabilities (\ref{transitionerlang}) do not depend on $\tau_{B}$ hence the steady-state probabilities obtained from (\ref{transitionerlang}) also do not depend on $\tau_{B}$. This leads to a property of $\tau_{B}$ which is shown in Theorem \ref{fixedpthreshold}. 
The unit-battery case , i.e., $B=1$ case was solved in ~\cite{WuYang2017} and  ~\cite{TanISIT2017}, hence we skip the case $B=1$ and continue with the case $B=2$ 
where we can show the result below:
\begin{theorem}
\label{B2age}
When $B=2$, the average age $\bar{\Delta}$ can be expressed as:
\begin{eqnarray}
&\bar{\Delta}=&\nonumber \\
&\frac{\!\frac{\alpha_{2}^{2}}{2}\!+\! e^{-\alpha_{2}}\![\alpha_{2}+1+\rho_{1}(\alpha_{2}^{2}+2\alpha_{2}+2)]\! -e^{-\alpha_{1}}\![\alpha_{1}+1+\rho_{1}(\alpha_{1}^{2}+\alpha_{1}+1)]}{\mu_{H}\left( \alpha_{2}+e^{-\alpha_{2}}[1+\rho_{1}(\alpha_{2}+1)]-e^{-\alpha_{1}}[1+\rho_{1}\alpha_{1}]\right) }&
\label{eq:B2age}
\end{eqnarray}
where
\[
\rho_{1}=\frac{e^{-\alpha_{1}}}{1-e^{-\alpha_{1}}\alpha_{1}},
\]
and
\[
\alpha_{1}=\mu_{H}\tau_{1}, \alpha_{2}=\mu_{H}\tau_{2}.
\]

\end{theorem}

\section{NUMERICAL RESULTS}
For battery sizes $B=1,2,3,4,5$, the policies in $\Pi^{\rm{MT}}$ are numerically optimized giving AoI versus energy arrival rate (Poisson) curves in Fig \ref{AgeEnergy}.
\begin{figure}[htpb]
\centering
  \begin{psfrags}
    \psfrag{Y}[t]{$\bar{\Delta}$}
    \psfrag{X}[b]{$\mu_{H}$}
    \psfrag{a}[l]{$B=1$}
    \psfrag{b}[l]{$B=2$}
    \psfrag{c}[l]{$B=3$}
    \psfrag{d}[l]{$B=4$}
    \psfrag{e}[l]{$B=5$}
    \psfrag{f}[l]{$B=\infty$}
   \includegraphics[scale=0.19]{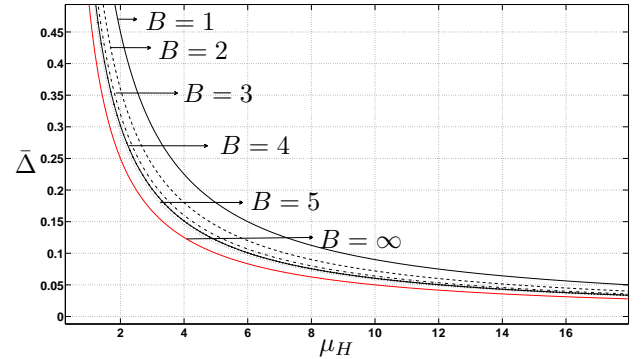}
    \end{psfrags}
\caption{AoI versus energy arrival rate (Poisson) for different battery sizes $B=1,2,3,4,5$.}
\label{AgeEnergy} 
\end{figure}

\begin{table}[]
\centering
\caption{Optimal thresholds for different battery sizes for $\mu=1$}
\label{my-label}
\begin{tabular}{|l|l|l|l|l|l|l|}
\hline
    & $\tau_{1}$ & $\tau_{2}$     & $\tau_{3}$     & $\tau_{4}$      & $\tau_{5}$  &$\bar{\Delta}_{\pi^*} $    \\ \hline
$B=1$ & 0.90         & - &  -    &   -    &   -    & 0.90 \\ \hline    
$B=2$ & 1.5         & 0.72 &  -    &   -    &   -    & 0.72 \\ \hline
$B=3$ & 1.5         & 1.2  & 0.64 &   -    &    -   & 0.64\\ \hline
$B=4$ & 1.5         & 1.2  & 0.96 & 0.604 &   -    & 0.604\\ \hline
$B=5$ & 1.5         & 1.2  & 0.96 & 0.9   & 0.582  & 0.582\\ \hline
\end{tabular}
\end{table}

\vspace{-0.1 in}
\section{CONCLUSION}
This paper explored the age-energy tradeoff for status updates sent by a finite-battery source that is charged intermittently by Poisson energy arrivals. The objective was to design a policy for the source to send updates to minimizing average status age using the given energy harvests, known and used in an online manner. A threshold policy is one that transmits when age exceeds a particular threshold for any battery state. It is shown that there is an online energy-causal threshold policy with monotone thresholds that optimally solves the problem. In particular, the smallest of the thresholds, the one used when the battery is full, has a value that matches the optimal average age.
\vspace{-0.1 in}
\bibliography{AgeOfInformation}

\begin{thebibliography}{10}
\providecommand{\url}[1]{#1}
\csname url@samestyle\endcsname
\providecommand{\newblock}{\relax}
\providecommand{\bibinfo}[2]{#2}
\providecommand{\BIBentrySTDinterwordspacing}{\spaceskip=0pt\relax}
\providecommand{\BIBentryALTinterwordstretchfactor}{4}
\providecommand{\BIBentryALTinterwordspacing}{\spaceskip=\fontdimen2\font plus
\BIBentryALTinterwordstretchfactor\fontdimen3\font minus
  \fontdimen4\font\relax}
\providecommand{\BIBforeignlanguage}[2]{{%
\expandafter\ifx\csname l@#1\endcsname\relax
\typeout{** WARNING: IEEEtran.bst: No hyphenation pattern has been}%
\typeout{** loaded for the language `#1'. Using the pattern for}%
\typeout{** the default language instead.}%
\else
\language=\csname l@#1\endcsname
\fi
#2}}
\providecommand{\BIBdecl}{\relax}
\BIBdecl

\bibitem{Kaul2011}
S.~Kaul, M.~Gruteser, V.~Rai, and J.~Kenney, ``Minimizing age of information in
  vehicular networks,'' in \emph{Sensor, Mesh and Ad Hoc Communications and
  Networks (SECON), 2011 8th Annual IEEE Communications Society Conference on},
  June 2011, pp. 350--358.

\bibitem{Kaul2012}
S.~Kaul, R.~Yates, and M.~Gruteser, ``Real-time status: How often should one
  update?'' in \emph{INFOCOM 2012}, pp. 2731--2735.

\bibitem{ZviedrisESMS10}
R.~Zviedris, A.~Elsts, G.~Strazdins, A.~Mednis, and L.~Selavo, ``Lynxnet: Wild
  animal monitoring using sensor networks,'' in \emph{REALWSN 2010}, 2010, pp.
  170--173.

\bibitem{blueforce}
K.~R. Chevli, P.~Kim, A.~Kagel, D.~Moy, R.~Pattay, R.~Nichols, and A.~D.
  Goldfinger, ``Blue force tracking network modeling and simulation,'' in
  \emph{MILCOM 2006}, Oct 2006, pp. 1--7.

\bibitem{Ephremides2013}
C.~Kam, S.~Kompella, and A.~Ephremides, ``Age of information under random
  updates,'' in \emph{IEEE ISIT}, July 2013, pp. 66--70.

\bibitem{Ephremides2014}
M.~Costa, M.~Codreanu, and A.~Ephremides, ``Age of information with packet
  management,'' in \emph{IEEE ISIT}, June 2014, pp. 1583--1587.

\bibitem{Huang2015}
L.~Huang and E.~Modiano, ``Optimizing age-of-information in a multi-class
  queueing system,'' in \emph{IEEE ISIT}, June 2015, pp. 1681--1685.

\bibitem{Pappas2015}
N.~Pappas, J.~Gunnarsson, L.~Kratz, M.~Kountouris, and V.~Angelakis, ``Age of
  information of multiple sources with queue management,'' in \emph{2015 ICC},
  June 2015, pp. 5935--5940.

\bibitem{Ephremides2016}
C.~Kam, S.~Kompella, G.~D. Nguyen, and A.~Ephremides, ``Effect of message
  transmission path diversity on status age,'' \emph{IEEE Transactions on
  Information Theory}, vol.~62, no.~3, pp. 1360--1374, March 2016.

\bibitem{Najm2016}
E.~Najm and R.~Nasser, ``Age of information: The gamma awakening,'' in
  \emph{IEEE ISIT}, July 2016, pp. 2574--2578.

\bibitem{DBLP:journals/corr/YatesK16}
\BIBentryALTinterwordspacing
R.~D. Yates and S.~K. Kaul, ``The age of information: Real-time status updating
  by multiple sources,'' \emph{CoRR}, vol. abs/1608.08622, 2016. [Online].
  Available: \url{http://arxiv.org/abs/1608.08622}
\BIBentrySTDinterwordspacing

\bibitem{StatusUpdateHARQ}
E.~Najm, R.~Yates, and E.~Soljanin, ``Status updates through m/g/1/1 queues
  with harq,'' in \emph{2017 IEEE International Symposium on Information Theory
  (ISIT)}, June 2017, pp. 131--135.

\bibitem{YinSunInfocom2016}
Y.~Sun, E.~Uysal-Biyikoglu, R.~Yates, C.~E. Koksal, and N.~B. Shroff, ``Update
  or wait: How to keep your data fresh,'' in \emph{IEEE INFOCOM 2016}, April
  2016, pp. 1--9.

\bibitem{YinIT2017}
Y.~Sun, E.~Uysal-Biyikoglu, R.~D. Yates, C.~E. Koksal, and N.~B. Shroff,
  ``Update or wait: How to keep your data fresh,'' \emph{IEEE Transactions on
  Information Theory}, vol.~63, no.~11, pp. 7492--7508, Nov 2017.

\bibitem{TanITA2015}
T.~Bacinoglu, E.~T. Ceran, and E.~Uysal-Biyikoglu, ``Age of information under
  energy replenishment constraints,'' in \emph{Proc.~Info. Theory and Appl.
  Workshop}, Feb. 2015.

\bibitem{2015ISITYates}
R.~D. Yates, ``Lazy is timely: Status updates by an energy harvesting source,''
  2015.

\bibitem{TanISIT2017}
T.~Bacinoglu and E.~Uysal-Biyikoglu, ``Scheduling status updates to minimize
  age of information with an energy harvesting sensor,'' in
  \emph{Proc.International Symp. on ~Info. Theory (ISIT)}, Jun. 2017.

\bibitem{WuYang2017}
X.~Wu, J.~Yang, and J.~Wu, ``Optimal status update for age of information
  minimization with an energy harvesting source,'' \emph{IEEE Transactions on
  Green Communications and Networking}, vol.~PP, no.~99, pp. 1--1, 2017.

\bibitem{twohop_energyharvesting}
A.~Arafa and S.~Ulukus, ``Age-minimal transmission in energy harvesting two-hop
  networks,'' Apr 2017.

\bibitem{ArafaDelay}
------, ``Age minimization in energy harvesting communications:
  Energy-controlled delays,'' Dec 2017.

\bibitem{peskir2006optimal}
G.~Peskir and A.~Shiryaev, \emph{Optimal Stopping and Free-Boundary Problems},
  ser. Lectures in Mathematics. ETH Z{\"u}rich.\hskip 1em plus 0.5em minus
  0.4em\relax Birkh{\"a}user Basel, 2006.

\bibitem{gallager2013stochastic}
R.~Gallager, \emph{Stochastic Processes: Theory for Applications}.\hskip 1em
  plus 0.5em minus 0.4em\relax Cambridge University Press, 2013.

\end{thebibliography}

\appendix
\subsection{The Proof of Theorem \ref{existopthreshold}}
Consider the problem in below for some $h<\infty$:
\begin{equation}
\label{kfinite}
\displaystyle\min_{\pi\in\Pi} \mathbb {E}\left[ \displaystyle\int_{Z_{k}}^{Z_{k+1}+h}\Delta(t)  dt \middle| Z_{k}=z,\mathcal{F}_{z}\right].
\end{equation}
In order to solve (\ref{kfinite}), let us define a cost function $J_{h;w,\mathcal{F}_{w}}$ for some time $w\geq z$ which is defined as:
\begin{align}
&J_{h;w,\mathcal{F}_{w}}^{*}:=&\nonumber\\
&\displaystyle\min_{\pi\in\Pi}\mathbb {E}\left[ \displaystyle\int_{w}^{Z_{k+1}+h}\Delta(t)  dt \middle|Z_{k+1} \geq w, \mathcal{F}_{w} \right], N_U(w^{-})=k.& 
\end{align}

This represents the minimum cumulative age in $[w,w+h]$ that is achievable by online policies given  $\mathcal{F}_{w}$. In fact,

\begin{lemma}
\label{starcost}
The cost function $J_{h;w,\mathcal{F}_{w}}^{*}$ depends only on $\Delta(w)$ and $E(w)$ , i.e., $J_{h;w,\mathcal{F}_{w}}^{*}=J_{h;w',\mathcal{F}_{w'}}^{*}$ if and only if $\Delta(w)=\Delta(w')$ and $E(w)=E(w')$.
\end{lemma}
\begin{proof}
This is due to the following facts that , for any $Z_{k+1}\geq w$, (1) given $\Delta(w)$ the cumulative age in $[w,Z_{k+1}+h]$, i.e. $\int_{w}^{Z_{k+1}+h}\Delta(t)  dt $ depends only the information on  the updates in $[w,Z_{k+1}+h]$ and (2) given $E(w)$, the evolution of the battery state in $[w,Z_{k+1}+h]$ is determined only by the updates and energy arrivals in $[w,Z_{k+1}+h]$ and the distribution of energy arrivals is identical for any $[w,Z_{k+1}+h]$.
\end{proof}

Hence, we can use the notation $J_{h}^{*}(a,\ell):=J_{h;w,\mathcal{F}_{w}}^{*}$ where $a=\Delta(w^+)$ and $\ell=E(w^+)$. Now, considering the case $Z_{k+1}=w$, define the following function:
\begin{align}
&J_{h}:=&\nonumber\\
&\displaystyle\min_{\pi\in\Pi}\mathbb {E}\left[ \displaystyle\int_{w}^{Z_{k+1}+h}\Delta(t)  dt \middle|Z_{k+1} = w, \mathcal{F}_{w} \right], N_U(w^{-})=k.& 
\end{align}

Notice that $\Delta(w^+)=0$ as $w=Z_{k+1}$, accordingly, by Lemma \ref{starcost}, $J_{h}$ is only a function of $E(w^+)$, i.e., $J_{h}=J_{h}(\ell)$ where $\ell=E(w^+)$. Notice also,
\begin{equation}
\mathbb {E}\left[ \displaystyle\int_{Z_{k+1}}^{Z_{k+1}+h}\Delta(t)  dt \middle| Z_{k+1},\mathcal{F}_{Z_{k+1}}\right]\geq J_{h}(E(Z_{k+1})-1),
\end{equation}
for any $Z_{k+1}$ of a policy $\pi \in \Pi$ and the equality is achieved for the policies that solve (\ref{kfinite}).

Accordingly, $Z_{k+1}$ of the policy solving (\ref{kfinite}) is the optimal stopping time of the following stopping problem for a given $z$, $h$ and $\mathcal{F}_{z}$:
\begin{equation}
\label{stoppingae}
\displaystyle\max_{w \in \mathfrak{M}_{z}} \mathbb {E}\left[ G_{w} \mid Z_{k}=z,\mathcal{F}_{z}\right],
\end{equation}
where $\mathfrak{M}_{z}$ is the family of stopping times such that 
$\mathfrak{M}_{z}=\left\lbrace w \geq z : \left\lbrace w \leq t\right\rbrace \in \mathcal{F}_{t}, \forall t \geq z\right\rbrace $  and $G=(G_{t})_{t \geq z}$ is a stochastic process having the following definition:
\begin{align}
\label{gainalt}
&G_{t}=&\nonumber\\
&-\displaystyle\min_{\pi\in\Pi}\mathbb {E}\left[ \displaystyle\int_{z}^{Z_{k+1}+h}\Delta(t)  dt \middle|Z_{k+1}=t, E(t)\right].& 
\end{align}
or alternatively,
\begin{equation}
G_{t}=-\frac{1}{2}(t-z)^{2}-J_{h}(E(t)-1),
\end{equation}
where $J_{h}(-1):=\infty$.

If exists, the optimal stopping time $w^{*}$ for (\ref{stoppingae}) is given by the following stopping rule  \cite[Theorem 2.2.]{peskir2006optimal}:
\begin{equation}
\label{rulesnell}
w^{*} = \inf \lbrace w \geq z : G_{w} = S_{w} \rbrace,
\end{equation}
where $S$ is the Snell envelope \cite{peskir2006optimal} for $G$:
\begin{equation}
\label{snellsw}
S_{w}= \esssup_{w' \in \mathfrak{M}_{w}}  \mathbb {E}\left[ G_{w'} \mid\mathcal{F}_{w}\right].
\end{equation}
Notice that the Snell envelope can be  written by substituting (\ref{gainalt}) in (\ref{snellsw})  as follows:
\begin{align}
&S_{w}=&\nonumber\\
&\esssup_{w' \in \mathfrak{M}_{w}}\left( -\displaystyle\min_{\pi\in\Pi}\mathbb {E}\left[ \displaystyle\int_{z}^{Z_{k+1}+h}\Delta(t)  dt \middle|Z_{k+1}=w', \mathcal{F}_{w}\right]\right) .& 
\end{align}
hence,
\begin{align}
&S_{w}=&\nonumber\\
&-\displaystyle\min_{\pi\in\Pi}\mathbb {E}\left[ \displaystyle\int_{z}^{Z_{k+1}+h}\Delta(t)  dt \middle|Z_{k+1}\geq w, \mathcal{F}_{w}\right].& 
\end{align}

Accordingly, using the definition of $J_{h}^{*}(a,\ell)$, we can write:
\begin{equation}
S_{w}=-\frac{1}{2}(w-z)^{2}-J_{h}^{*}(w-z,E(w)).
\end{equation}

Therefore, the optimal stopping rule in (\ref{rulesnell}) is equivalent to:
\begin{equation}
\label{ruleage}
w^{*} = \inf \lbrace w \geq z : J_{h}(E(w)-1) = J_{h}^{*}(\Delta(w),E(w))\rbrace,
\end{equation}
Next, we show that the stopping rule in (\ref{ruleage}) is a threshold rule in age. In order to show this, let us define the function $\rho_{h}(\cdot):\lbrace -1, 0,1,...,B\rbrace \rightarrow [0,\infty)$ such that:
\[
\rho_{h}(\ell)=\inf\lbrace a \geq 0 : J_{h}(\ell-1) = J_{h}^{*}(a,\ell)\rbrace.
\]
Consider $J_{h}^{*}(a',\ell)$ for some $a'\geq \rho_{h}(\ell)$ which is larger than or equal to  $J_{h}^{*}(a,\ell)$ as  $J_{h}^{*}(a,\ell)$ is non-decreasing in $a$. On the other hand, $J_{h}^{*}(a',\ell)$ is smaller than or equal to $J_{h}(\ell-1)$ for any $a \geq 0$ as:
\begin{align*}
&J_{h}^{*}(a',\ell) &\\
&\leq\displaystyle\min_{\pi\in\Pi}\mathbb {E}\left[ \displaystyle\int_{w}^{Z_{k+1}+h}\Delta(t)  dt \middle| Z_{k+1}=w, E(w)=\ell \right]&\\
&=J_{h}(\ell-1),&
\end{align*}
where the inequality is true as the expectation is conditioned on policies with $Z_{k+1}=w$.

Accordingly, $J_{h}^{*}(a',\ell)=J_{h}(\ell-1)$ for any $\ell \in \lbrace 0, 1,2,.., B\rbrace$ and $a' \geq \rho_{h}(\ell)$. Therefore, the stopping rule in (\ref{ruleage}) is equivalent to:
\begin{equation}
\label{ruleaged}
w^{*} = \inf \lbrace w \geq z : \Delta(w)\geq \rho_{h}(E(w))\rbrace,
\end{equation}
for $\ell \in \lbrace 0, 1,2,.., B\rbrace$.

We showed that the stopping rule in (\ref{ruleaged}) gives the optimal stopping time $w^{*}$ which equals to $Z_{k+1}$ of a policy solving  (\ref{kfinite}) for any finite $h$. Now, we show that the optimal stopping rule with the same structure also gives a solution to (\ref{minavaged}). 

First, consider:
\begin{lemma}
The function $\rho_{h}(\ell)$ is uniformly bounded such that:
\begin{equation}
\rho_{max}=\sup_{h \geq 0 ,\ell>0} \rho_{h}(\ell).
\end{equation}
\end{lemma}
We will show that this lemma implies $\Pr(X_{k}\geq x) \in O(e^{-\mu_H x})$ when $Z_{k+1}=w^{*}$. This follows from the fact that, when $Z_{k+1}=w^{*}$, $X_{k}\geq \rho_{max}+x_{d}$ for some $x_{d}>0$ is possible if and only if, for some time $t$, $E(t)=0$ and $\Delta(t)=\rho_{max}+x_{d}$,  which occurs when there is no energy arrival during $\rho_{max}+x_{d}$ units of time. This also shows that $\Pr(Z_{k+1}< \infty \mid Z_{k}=z)=1$.  

Next, we show that:
\begin{lemma}
When $Z_{k+1}=w^{*}$ ,
\begin{align}
\label{hlimit}
&\lim_{h\rightarrow \infty}\frac{1}{h}\mathbb {E}\left[ \displaystyle\int_{z}^{Z_{k+1}+h}\Delta(t)  dt \middle| Z_{k}=z,\mathcal{F}_{z}\right]&\nonumber\\
&= \lim_{h\rightarrow \infty}\frac{1}{h}\mathbb {E}\left[ \displaystyle\int_{z}^{h}\Delta(t)  dt \middle| Z_{k}=z,\mathcal{F}_{z}\right],& 
\end{align}
\end{lemma}
\begin{proof}
 This equality can be shown considering: (i) the case ($\leq$) and (ii) the case ($\geq$):

(i) For the case ($\leq$), consider:
\begin{align*}
&\mathbb {E}\left[ \displaystyle\int_{Z_{k}}^{Z_{k+1}+h}\!\!\Delta(t)  dt \middle|Z_{k},\mathcal{F}_{Z_{k}}\right] &\\
&=\mathbb {E}\left[ \displaystyle\int_{Z_{k}}^{Z_{k+1}+h}\!\!\Delta(t)  dt \middle| Z_{k},\mathcal{F}_{Z_{k}},Z_{k+1}< h_{\alpha} \right]\!\Pr(Z_{k+1}< h_{\alpha})&\\
&+\mathbb {E}\left[ \displaystyle\int_{Z_{k}}^{Z_{k+1}+h}\!\!\Delta(t)  dt \middle| Z_{k},\mathcal{F}_{Z_{k}},Z_{k+1}\geq h_{\alpha} \right]\!\Pr(Z_{k+1}\geq h_{\alpha})&\\
&\leq\mathbb {E}\left[ \displaystyle\int_{Z_{k}}^{h(\alpha+1)}\!\!\Delta(t)  dt \middle|Z_{k},\mathcal{F}_{Z_{k}},Z_{k+1}<h_{\alpha} \right]\!\Pr(Z_{k+1}< h_{\alpha})&\\
&+\mathbb {E}\left[ \displaystyle\int_{Z_{k}}^{Z_{k+1}+h}\!\!\Delta(t)  dt \middle| Z_{k},\mathcal{F}_{Z_{k}},Z_{k+1}\geq h_{\alpha} \right]\!\Pr(Z_{k+1}\geq h_{\alpha}),&
\end{align*}
where $h_{\alpha}=\alpha h$ for some $\alpha \in (0,1)$.

The term for the condition $Z_{k+1}\geq h_{\alpha}$ vanishes as $h\rightarrow\infty$, in order to see this consider:
\begin{align*}
&\mathbb {E}\left[ \displaystyle\int_{Z_{k}}^{Z_{k+1}+h}\Delta(t)  dt  \middle| Z_{k},\mathcal{F}_{Z_{k}},Z_{k+1}\geq h_{\alpha} \right]\Pr(Z_{k+1}\geq h_{\alpha})&\\
&\leq\frac{1}{2}\left( \mathbb {E}\left[ X_{k}^{2} \mid X_{k}\geq h_{\alpha}- z\right]+h^{2}\right) \Pr(X_{k}\geq h_{\alpha}-z).&
\end{align*}
For $Z_{k+1}=w^{*}$, the upper bound goes to zero when $h\rightarrow\infty$ as $\Pr(X_{k}\geq h_{\alpha}-z) \in O(e^{-\mu_H h_{\alpha}})$ and consequently $ \mathbb {E}\left[ X_{k}^{2} \mid X_{k}\geq h_{\alpha}- z\right] \in O(h_{\alpha}^{2})$.

Accordingly,

\begin{align}
\label{upperboundforkfinitealpha}
&\lim_{h\rightarrow \infty}\frac{1}{h}\mathbb {E}\left[ \displaystyle\int_{z}^{Z_{k+1}+h}\Delta(t)  dt \middle| Z_{k}=z,\mathcal{F}_{z}\right]&\nonumber\\
&\leq \lim_{h\rightarrow \infty}\frac{1}{h}\mathbb {E}\left[ \displaystyle\int_{z}^{h(\alpha+1)}\Delta(t)  dt \middle| Z_{k}=z,\mathcal{F}_{z}\right], \alpha \in (0,1).& 
\end{align}

As the inequality is true for any $\alpha \in (0,1)$:
\begin{align}
\label{upperboundforkfinitealphaop}
&\lim_{h\rightarrow \infty}\frac{1}{h}\mathbb {E}\left[ \displaystyle\int_{z}^{Z_{k+1}+h}\Delta(t)  dt \middle| Z_{k}=z,\mathcal{F}_{z}\right]&\nonumber\\
&\leq \lim_{h\rightarrow \infty}\frac{1}{h}\mathbb {E}\left[ \displaystyle\int_{z}^{h}\Delta(t)  dt \mid Z_{k}=z,\mathcal{F}_{z}\right].& 
\end{align}
(ii) For the case ($\geq$), it can be seen that:
\begin{align*}
\label{lowerboundforkfinite}
&\lim_{h\rightarrow \infty}\frac{1}{h}\mathbb {E}\left[ \displaystyle\int_{z}^{Z_{k+1}+h}\Delta(t)  dt \middle| Z_{k}=z,\mathcal{F}_{z}\right]&\nonumber\\
&\geq \lim_{h\rightarrow \infty}\frac{1}{h}\mathbb {E}\left[ \displaystyle\int_{z}^{h}\Delta(t)  dt \middle| Z_{k}=z,\mathcal{F}_{z}\right].& 
\end{align*}
\end{proof}
Therefore, (\ref{hlimit}) is true and by (\ref{hlimit}) and (\ref{ruleaged}), a solution to the following problem,
\begin{equation}
\label{zproblem}
\displaystyle\min_{\pi\in\Pi}\lim_{h\rightarrow \infty}\frac{1}{h}\mathbb {E}\left[ \displaystyle\int_{z}^{h}\Delta(t)  dt \middle| Z_{k}=z,\mathcal{F}_{z}\right],
\end{equation} 
satisfies:
\begin{equation}
\label{ruleagehinf}
Z_{k+1} = \inf \lbrace w \geq z : \Delta(w)\geq \rho(E(w))\rbrace,
\end{equation}
where $\rho(\ell) := \lim_{h\rightarrow \infty} \rho_{h}(\ell)$.

Notice that:
\begin{align*}
&\mathbb {E}\left[ \displaystyle\int_{0}^{t_{f}}\Delta(t)  dt \right]&\\
&=\mathbb {E}\left[ \mathbb {E}\left[ \displaystyle\int_{0}^{Z_{k}}\Delta(t)  dt\middle|Z_{k},\mathcal{F}_{Z_{k}}\right] +\mathbb {E}\left[ \int_{Z_{k}}^{t_{f}}\Delta(t)\middle|Z_{k},\mathcal{F}_{Z_{k}}\right]\right].&
\end{align*}
 Therefore, conditioned on $Z_{k}$ and $\mathcal{F}_{Z_{k}}$, minimizing $\bar{\Delta}=\lim\sup_{t_{f} \rightarrow \infty}\frac{1}{t_{f}}\mathbb {E}\left[ \displaystyle\int_{0}^{t_{f}}\Delta(t)  dt \right]$ corresponds to solving (\ref{zproblem}), which means  a threshold policy solves (\ref{minavaged}) as (\ref{ruleagehinf}) (where $\rho(\ell)=\tau_{\ell}$) is satisfied by a solution to (\ref{zproblem}).

\subsection{The Proof of Theorem \ref{existopmothreshold}}
Theorem \ref{existopmothreshold} follows from the proof of Theorem \ref{existopthreshold} and the following lemma:
\begin{lemma}
$J_{h}(\ell)-J_{h}^{*}(a,\ell+1)$ is non-increasing in $\ell \in \lbrace 0,1,...,B-1\rbrace$ for any $a\geq0$ and $h \geq 0$.
\end{lemma}
\begin{proof}
First, consider the problem in below for some $\ell \in \lbrace 0,1,...,B-1\rbrace$ and and $h \geq 0$:
\begin{align}
\label{energycondition}
&\min_{V} \mathbb {E}_{V}\left[\min_{\pi \in \Pi} \mathbb {E}\left[ \displaystyle\int_{Z_{k+1}}^{Z_{k+1}+h}\Delta(t)  dt \middle| Z_{k+1}, E(Z_{k+1})=\ell+V\right]\right]&\\
&\mbox{s.t. } \mathbb {E}[V]= 1.&
\end{align}
where $V$ is a discrete r.v. that takes values in $\lbrace 0,1, ..., B-\ell-1\rbrace$.

The problem in (\ref{energycondition}) is solved when $\Pr(V=1)=1$ as this case maximizes the prior knowledge on $V$. Accordingly, considering the case $\Pr(V=0)=\Pr(V=2)=\frac{1}{2}$ which is suboptimal in (\ref{energycondition}), it can be seen that $J_h(\ell)$ constitutes  a convex series in $\ell$:
\begin{equation}
\label{Jhconvexity}
J_h(\ell+1)\leq \frac{1}{2}(J_h(\ell)+J_h(\ell+2)),
\end{equation}
for any  $\ell \in \lbrace 0,1,...,B-1\rbrace$ and $h \geq 0$.

Now, consider the alternative formulation of $J_{h}^{*}(a,\ell+1)$ in below:
\begin{align}
\label{Jhal}
&J_{h}^{*}(a,\ell+1)=\min_{Z_{k+1} \in \mathfrak{M}_{w}} \sum_{\sigma=0}^{\infty}\int_{w}^{\infty} K(w',\sigma)\times& \nonumber\\ 
& \left[ (z'-w)(a+\frac{z'-w}{2})+J_{h}(\min\lbrace\ell+\sigma,B-1\rbrace)\right]dz',&
\end{align}
where $K(w',\sigma)=\Pr(Z_{k+1}=z', N_H(z')-N_H(w)=\sigma )$.

Similarly,
\begin{align}
\label{decreasedJhal}
&J_{h}^{*}(a,\ell+2)=\min_{Z_{k+1} \in \mathfrak{M}_{w}} \sum_{\sigma=0}^{\infty}\int_{w}^{\infty}K(w',\sigma)\times& \nonumber\\ 
& \left[ (z'-w)(a+\frac{z'-w}{2})+J_{h}(\min\lbrace\ell+1+\sigma,B-1\rbrace)\right]dz'.&
\end{align}
Now, let $K^{*}(w',\sigma)$ be the distribution corresponding to the update time $Z_{k+1} \in \mathfrak{M}_{w}$ that is optimal in (\ref{decreasedJhal}), which means:
\begin{align}
\label{decreasedJhalop}
&J_{h}^{*}(a,\ell+2)=\sum_{\sigma=0}^{\infty}\int_{w}^{\infty}K^{*}(w',\sigma)\times& \nonumber\\ 
&\left[ (z'-w)(a+\frac{z'-w}{2})+J_{h}(\min\lbrace\ell+1+\sigma,B-1\rbrace)\right]dz'.& 
\end{align}
Combining (\ref{decreasedJhalop}) and (\ref{Jhal}) gives:
\begin{align}
\label{decreasedJhalop}
&J_{h}^{*}(a,\ell+1)- J_{h}^{*}(a,\ell+2)\leq \sum_{\sigma=0}^{\infty}\int_{w}^{\infty}K^{*}(w',\sigma)\times& \nonumber\\ 
&[J_{h}(\min\lbrace\ell+\sigma,B-1\rbrace)\!-\!J_{h}(\min\lbrace\ell+1+\sigma,B-1\rbrace)]dz'.&
\end{align}
As $\sigma \geq 0$, the  below inequality holds due to (\ref{Jhconvexity}):
\begin{align}
\label{convJh}
&J_{h}(\min\lbrace\ell+\sigma,B-1\rbrace)-J_{h}(\min\lbrace\ell+1+\sigma,B-1\rbrace)\leq & \nonumber\\ 
&J_h(\ell)-J_h(\ell+1).&
\end{align}
Hence,
\begin{align}
\label{convJhal}
J_{h}^{*}(a,\ell+1)- J_{h}^{*}(a,\ell+2)\leq  
J_h(\ell)-J_h(\ell+1),
\end{align}
which means:
\begin{align}
\label{nondJhminJha}
J_h(\ell+1)- J_{h}^{*}(a,\ell+2)\leq  
J_h(\ell)-J_{h}^{*}(a,\ell+1).
\end{align}
\end{proof}
\vspace{-0.2in}
This lemma shows that $\rho_{h}(\ell)$ is non-increasing in $\ell$ for any $h \geq 0$ as:
\[
0=J_{h}(\ell-1)-J_{h}^{*}(\rho_{h}(\ell),\ell)\leq J_{h}(\ell-2)-J_{h}^{*}(\rho_{h}(\ell),\ell-1),
\]
hence $\rho_{h}(\ell-1)\geq \rho_{h}(\ell)$.

As $\tau_{\ell}=\lim_{h\rightarrow \infty} \rho_{h}(\ell)$ for an optimal policy and $\rho_{h}(\ell)$ is non-increasing, thus a policy with $\tau_B\leq \ldots\leq\tau_{\ell}\leq\tau_{1}$ solves (\ref{minavaged}).


\subsection{The proof of Lemma \ref{diffmoments}}
Taking $\tau_{B+1}=0$ and $\tau_{0}=\infty$, consider:
\begin{align*}
\frac{\partial}{\partial \tau_{i}}\mathbb {E}\left[ X^{2}\right]&=\frac{\partial}{\partial \tau_{i}} \int_{0}^{\infty}\Pr(X^{2}\geq x)dx\\
&=\frac{\partial}{\partial \tau_{i}}\displaystyle\sum_{i=0}^{B}\int_{\tau_{i+1}^{2}}^{\tau_{i}^{2}}\Pr(X\geq \sqrt{x})dx\\
&=\frac{\partial}{\partial \tau_{i}}[\int_{\tau_{i+1}^{2}}^{\tau_{i}^{2}}\Pr(X\geq \sqrt{x})dx\\
&+\int_{\tau_{i}^{2}}^{\tau_{i-1}^{2}}\Pr(X\geq \sqrt{x})dx],
\end{align*}
for any $i=0,1,...,B$.

Similarly,
\begin{align*}
\frac{\partial}{\partial \tau_{i}}\mathbb {E}\left[ X\right]&=\frac{\partial}{\partial \tau_{i}}[\int_{\tau_{i+1}}^{\tau_{i}}\Pr(X\geq x)dx\\
&+\int_{\tau_{i}}^{\tau_{i-1}}\Pr(X\geq x)dx].
\end{align*}
for $i=0,1,...,B$.

Let $\tilde{F}_{i}(x)=1-F_{i}(x)$ and $\tilde{F}_{i}^{I}(x)=\int_{0}^{x} \tilde{F}_{i}(x')dx'$.  Then,
\begin{align*}
&\int_{\tau_{i+1}^{2}}^{\tau_{i}^{2}}\Pr(X\geq \sqrt{x})dx+\int_{\tau_{i}^{2}}^{\tau_{i-1}^{2}}\Pr(X\geq \sqrt{x})dx=\\
&=2\tau_{i}\tilde{F}_{i}^{I}(\tau_{i})-Q_{i}(\tau_{i}^{2})-2\tau_{i+1}\tilde{F}_{i}^{I}(\tau_{i+1})+Q_{i}(\tau_{i+1}^{2})+\\
&+2\tau_{i-1}\tilde{F}_{i-1}^{I}(\tau_{i-1})-Q_{i-1}(\tau_{i-1}^{2})-2\tau_{i}\tilde{F}_{i-1}^{I}(\tau_{i})+Q_{i-1}(\tau_{i}^{2}),
\end{align*}
where $Q_{i}(x)=\int_{0}^{x} \frac{\tilde{F}_{i}^{I}(\sqrt{x'})}{\sqrt{x'}}dx$.

Accordingly,
\begin{align*}
\frac{\partial}{\partial \tau_{i}}\mathbb {E}\left[ X^{2}\right]&=2\tau_{i}\tilde{F}_{i}(\tau_{i})+2\tilde{F}_{i}^{I}(\tau_{i})-\frac{\tilde{F}_{i}^{I}(\tau_{i})}{\tau_{i}}(2\tau_{i})\\
&-2\tau_{i}\tilde{F}_{i-1}(\tau_{i})-2\tilde{F}_{i-1}^{I}(\tau_{i})+\frac{\tilde{F}_{i-1}^{I}(\tau_{i})}{\tau_{i}}(2\tau_{i})\\
&=2\tau_{i}(\tilde{F}_{i}(\tau_{i})-\tilde{F}_{i-1}(\tau_{i}))\\
&=2\tau_{i}\frac{\partial}{\partial \tau_{i}}\mathbb {E}\left[ X\right].
\end{align*}
for $i=0,1,...,B$.
\begin{lemma}
The DTMC with the transition probabilities in (\ref{transitionerlang}) is ergodic for monotonic threshold policy where  $\tau_{1}$ is finite.
\label{ergodicitylemma}
\end{lemma}
\subsection{The Proof of Lemma \ref{ergodicitylemma}}
Consider an energy state $j$ in $[0,B-1]$. We will show that any other energy state $i$ is reachable from $j$ in at most $B-1$ steps with a positive probability. For $i\geq j$, the higher energy state $i$ is reachable from $j$ in one step with a positive probability as for $i=B-1$, $\Pr(Y_{B-j}\leq \tau_{B-1})$ is strictly positive and for $ j\leq i<B-1$:
\begin{align*}
&\Pr(Y_{1+i-j}\leq \tau_{i})- \Pr(Y_{2+i-j}\leq \tau_{i+1})\geq\\
&\Pr(Y_{1+i-j}\leq \tau_{i+1})- \Pr(Y_{2+i-j}\leq \tau_{i+1})>0,
\end{align*}

as $\tau_{i+1}\leq \tau_{i}$ and $i-j \geq 0$.

Similarly, the energy state $i=j-1$ for $j=1,...., B-1$ can be reached from $j$ with a probability $1-\Pr(Y_{1}\leq \tau_{j})$ which is stricly positive as $\tau_{j}$ is finite. This means that any state $i<j$ can be reached from $j$ in at most $B-1$ steps with a positive probability. 
\begin{lemma}
\label{ergoupdates}
For monotonic threshold policies with finite $\tau_{1}$,  the following is true:
\begin{equation}
\lim_{n\rightarrow +\infty}\frac{1}{n}\sum_{k=0}^{n}X_{k}=\sum_{j=0}^{B-1}\mathbb {E}\left[ X \mid j\right] \Pr(E=j) \mbox{\:\: w.p.1.}
\label{monotoneX}
\end{equation}
\begin{equation}
\lim_{n\rightarrow +\infty}\frac{1}{2n}\sum_{k=0}^{n}\mathbb {E}[X_{k}^{2}]=\frac{1}{2}\sum_{j=0}^{B-1} \mathbb {E}\left[ X^{2} \mid j\right] \Pr(E=j),
\label{monotoneC}
\end{equation}
where  $\Pr(E=j)$ is the steady-state probability for energy state $j$, $\mathbb {E}\left[ X \mid j\right] \triangleq \mathbb {E}\left[ X_{k} \mid E(Z_{k})=j\right]$ and $\mathbb {E}\left[ X^{2} \mid j\right] \triangleq \mathbb {E}\left[ X_{k}^{2} \mid E(Z_{k})=j\right]$.
\end{lemma}
\begin{proof}
Consider:
\[
\frac{1}{n}\sum_{k=0}^{n} X_{k}=\frac{1}{n}\sum_{j=0}^{B-1}\sum_{\substack{k\in [0,n]\\ E(Z_{k})=j }} X_{k}=\frac{1}{n}\sum_{j=0}^{B-1}\displaystyle\sum_{\ell=0}^{L_{j}} X_{\ell;j},
\]
where $L_{j}$ is the number of $k$s in $[0,n]$ such that $E(Z_{k})=j$ and $X_{\ell;j}$ is a r.v. with the CDF $\Pr(X_{\ell;j}\leq x)= \Pr(X_{\ell}\leq x \mid  E(Z_{\ell})=j)$.

Note that the sequence $X_{0;j}, X_{1;j}, ...,X_{L_{j};j}$ is i.i.d. for any $j$ and their mean is bounded as all thresholds are finite, hence:
\[
\displaystyle\lim_{L_{j}\rightarrow \infty} \frac{1}{L_{j}}\displaystyle\sum_{\ell=0}^{L_{j}} X_{\ell;j}= \mathbb {E}\left[ X \mid j\right], w.p.1.
\]
Due to the ergoditicity of $E(Z_{k})$s (Lemma \ref{ergodicitylemma}):
\[
\displaystyle\lim_{n\rightarrow \infty}\frac{L_{j}}{n}=\Pr(E=j), w.p.1.
\]
Therefore,
\begin{align*}
\displaystyle\lim_{n\rightarrow \infty}\frac{1}{n}\displaystyle\sum_{k=0}^{n} X_{k}&=\lim_{n\rightarrow \infty}\displaystyle\sum_{j=0}^{B-1}\frac{L_{j}}{n}(\frac{1}{L_{j}}\displaystyle\sum_{\ell=0}^{L_{j}} X_{\ell;j}),\\
&=\displaystyle\sum_{j=0}^{B-1} \mathbb {E}\left[ X \mid j\right] \Pr(E=j), w.p.1.
\end{align*}
Similarly,
\begin{align*}
\displaystyle\lim_{n\rightarrow \infty}\frac{1}{n}\displaystyle\sum_{k=0}^{n} \mathbb {E}[X_{k}^{2}]&=\lim_{n\rightarrow \infty}\displaystyle\sum_{j=0}^{B-1}\frac{L_{j}}{n}(\frac{1}{L_{j}}\displaystyle\sum_{\ell=0}^{L_{j}} X_{\ell;j}^{2})\\
&=\displaystyle\sum_{j=0}^{B-1} \mathbb {E}\left[ X^{2} \mid j\right] \Pr(E=j), w.p.1.
\end{align*}
\end{proof}
\subsection{The proof of Theorem \ref{fixedpthreshold}}
By Lemma \ref{convage} and Lemma \ref{ergoupdates}, the average age can be written as follows:
\[
\bar{\Delta}=\frac{\sum_{j=0}^{B-1} \mathbb {E}\left[ X^{2} \mid j\right] \Pr(E=j)}{2 \sum_{j=0}^{B-1} \mathbb {E}\left[ X \mid j\right] \Pr(E=j)}.
\]
Define $D(\tau)$ as in below:
\[
D(\tau)=\displaystyle\sum_{j=0}^{B-1} (\mathbb {E}\left[ X^{2} \mid j\right]- 2\bar{\Delta}_{B}^{M}\mathbb {E}\left[ X \mid j\right])\Pr(E=j),
\]
where $\bar{\Delta}_{B}^{M}=\displaystyle\min_{\pi\in\Pi}\bar{\Delta}$.

By Lemma \ref{diffmoments}:
\[
\frac{\partial}{\partial \tau_{B}}\mathbb {E}\left[ X^{2} \mid j\right]=2\tau_{B}\frac{\partial}{\partial \tau_{B}}\mathbb {E}\left[ X \mid j\right].
\]
Accordingly, as $\frac{\partial}{\partial \tau_{B}}\Pr(E=j)=0$, 
\[
\frac{\partial}{\partial \tau_{B}}D(\tau)=2 (\tau_{B}-\bar{\Delta}_{B}^{M})\displaystyle\sum_{j=0}^{B-1} \frac{\partial}{\partial \tau_{B}}\mathbb {E}\left[ X \mid j\right] \Pr(E=j),
\]
$\frac{\partial}{\partial \tau_{B}}D(\tau)$ can be also written as:
\[
\frac{\partial}{\partial \tau_{B}}D(\tau)=2 (\tau_{B}-\bar{\Delta}_{B}^{M})\frac{\partial \bar{X}}{\partial \tau_{B}},
\]
where $\bar{X}=\displaystyle\sum_{j=0}^{B-1} \mathbb {E}\left[ X \mid j\right] \Pr(E=j)$.

It can be seen that $\frac{\partial \bar{X}}{\partial \tau_{B}}\geq 0$ for any $\tau_{B} \geq 0$ which means $\frac{\partial}{\partial \tau_{B}}D(\tau)$ can only change its sign around $\tau_{B}=\bar{\Delta}_{B}^{M}$. As $D(\tau)\geq 0$ and $D(\arg\min_{\tau_{B}:\pi \in \Pi^{\rm{MT}}}\bar{\Delta})=0$ by its definition, for $\tau$ that achieves $\bar{\Delta}_{B}^{M}$, $\tau_{B}=\bar{\Delta}_{B}^{M}$.
\subsection{The Proof of Theorem \ref{B2age}}

By Lemma \ref{convage} and Lemma \ref{ergoupdates}, $\bar{\Delta}$ for $B=2$ is the following:
\begin{equation}
\label{B2delta}
\bar{\Delta}=\frac{1}{2}\frac{\mathbb {E}\left[ X^{2} \mid j=0\right]\Pr(E=0)+\mathbb {E}\left[ X^{2} \mid j=1\right]\Pr(E=1)}{\mathbb {E}\left[ X \mid j=0\right]\Pr(E=0)+\mathbb {E}\left[ X \mid j=1\right]\Pr(E=1)}.
\end{equation}
The probability of being in $E=1$, i.e. $\Pr(E=1)$ can be solved using:
\begin{equation}
\label{steadyE}
\Pr(E=1)=\sum_{j=0}^{1}\Pr(E(Z_{k+1})=1 \mid  E(Z_{k})=j)\Pr(E=j).
\end{equation}
Combining (\ref{steadyE}) and (\ref{cdferlang}),
\begin{equation}
\label{steadyexp}
\Pr(E=1)=\frac{e^{-\mu_{H}\tau_{1}}}{1-e^{-\mu_{H}\tau_{1}}\mu_{H}\tau_{1}}.
\end{equation}
Now, we can obtain $\mathbb {E}\left[ X^{2} \mid j\right]$, $\mathbb {E}\left[ X \mid j\right]$ using (\ref{cdferlang}). Combining these with (\ref{steadyexp}) and substituting in (\ref{B2delta}) gives  (\ref{B2age}).

\begin{lemma}
\label{convage}
For a threshold policy where $\tau_{1}$ is finite, the average age $\bar{\Delta}$ is finite (w.p.1.) and given by the following expression.
\begin{equation}
\bar{\Delta}=\frac{\lim_{n\rightarrow +\infty}\frac{1}{2n}\sum_{k=0}^{n}\mathbb {E}[X_{k}^{2}]}{\lim_{n\rightarrow +\infty}\frac{1}{n}\sum_{k=0}^{n}X_{k}} \mbox{\:\: w.p.1.}
\end{equation}
\end{lemma}
\begin{proof}
The proof is a generalization of Theorem 5.4.5 in \cite{gallager2013stochastic} for the case where $X_{k}$s are non-i.i.d. but the limits still exist (w.p.1.). When $X_{k}$s  are i.i.d. with $\mathbb {E}[X_{k}]< \infty$ and $\mathbb {E}[X_{k}^{2}]< \infty$, the convergence (w.p.1.) of the limits is guaranteed.
\end{proof}


\end{document}